\documentclass[%
 aip,
amsmath,amssymb,
 reprint,%
]{revtex4-1}

\usepackage{graphicx}% Include figure files
\usepackage{subcaption}
\usepackage{dcolumn}% Align table columns on decimal point
\usepackage{bm}% bold math
\usepackage{amsthm} % For proof environment

\usepackage[utf8]{inputenc}
\usepackage[T1]{fontenc}
\usepackage{mathptmx}
\usepackage{etoolbox}
\usepackage{hyperref}
\usepackage{cleveref}
\usepackage{bbm}
\usepackage{todonotes}

\newtheorem{definition}{Definition}[subsection]
\newtheorem{theorem}[definition]{Theorem}
\newtheorem{lemma}[definition]{Lemma}

\crefname{definition}{definition}{definitions}
\Crefname{definition}{Definition}{Definitions}
\crefname{theorem}{theorem}{theorems}
\Crefname{theorem}{Theorem}{Theorems}
\crefname{lemma}{lemma}{lemmas}
\Crefname{lemma}{Lemma}{Lemmas}

\makeatletter
\def\@email#1#2{%
 \endgroup
 \patchcmd{\titleblock@produce}
  {\frontmatter@RRAPformat}
  {\frontmatter@RRAPformat{\produce@RRAP{*#1\href{mailto:#2}{#2}}}\frontmatter@RRAPformat}
  {}{}
}%
\makeatother
\begin{document}

\preprint{AIP/123-QED}

\title[Chaotic Kramers’ Law]{Chaotic Kramers’ Law: Hasselmann’s Program and AMOC Tipping}

\author{J. Deser}
\email{jakob.deser@tum.de}
\affiliation{Technical University of Munich, Department of Mathematics, 85748~Garching bei M\"unchen, Germany}
 
\author{R. Römer}
\email{r.k.roemer@exeter.ac.uk}
\affiliation{Department of Mathematics and Statistics, University of Exeter, Exeter, EX4 4QF, UK}

\author{N. Boers}
\affiliation{Earth System Modelling, School of Engineering and Design, Technical University of Munich, Munich 85521, Germany}
\affiliation{Potsdam Institute for Climate Impact Research, Potsdam 14473, Germany}
\affiliation{Department of Mathematics and Global Systems Institute, University of Exeter, Exeter EX4 4SB, UK}

\author{C.Kuehn}
\affiliation{Technical University of Munich, School of Computation, Information and Technology, Department of Mathematics, 85748~Garching bei M\"unchen, Germany}
\affiliation{Technical University of Munich, Munich Data Science Institute (MDSI), 85748 Garching bei M\"unchen, Germany}

\date{\today}

\begin{abstract}
In bistable dynamical systems driven by Wiener processes, the widely used Kramers' law relates the strength of the noise forcing to the average time it takes to see a noise-induced transition from one attractor to the other. We extend this law to bistable systems forced by fast chaotic dynamics, which we argue is in some cases a more realistic modeling approach than unbounded noise forcing. Transitions similar to the noise-driven case can only occur if the amplitude of the chaotic forcing is large enough. If this is the case, in our numerical example - a reduced-order model of the Atlantic Meridional Overturning Circulation (AMOC) - we observe the chaotic Kramers' law to hold even when the chaotic forcing is far from the stochastic limit. We discuss the limitations of the chaotic Kramers' law, how to address the numerical issues associated with the timescale separation, and give a possible explanation for the dynamics of recently found AMOC collapses and recoveries in complex climate models.
\end{abstract}

\maketitle

\begin{quotation}
A common approach when modeling dynamical systems is to separate a system into slow and fast components, and approximate the fast and often chaotic dynamics by a Wiener process that acts as forcing on the slow dynamics. One can ask what happens if the fast dynamics are not quite fast enough for this stochastic approximation to hold, but when, instead, chaotic forcing would be more realistic. For such cases, we derive the chaotic Kramers' law that relates the strength of the chaotic forcing to tipping times between two attractors of the forced slow dynamics. We explain key assumptions for the chaotic Kramers' law to hold, its surprisingly wide applicability, and its limitations while illustrating and interpreting our findings in a chaotically forced reduced-order AMOC model.
\end{quotation}

\section{Introduction}

The Hasselmann program from the 1970s \cite{Hasselmann_1976} is a widely used modeling approach in climate science, and similar techniques are commonly used in many other fields \cite{berglund2006noise,kuehn2015multiple}. The program proposed to model fast, often chaotic, fluctuations as noise, forcing the slower dynamics. This approach allows one to efficiently model and understand the slow dynamics and many large-scale processes, which would prove difficult if one tried to resolve all small-scale, fast processes in detail.

However, in many cases, it is possible - and if one is interested in long or many model runs, even necessary -  to describe natural processes with reduced-order models. In these simpler models, it can be feasible to also resolve the faster dynamics, which leads to the question of whether, in these models, a stochastic approximation of the fast dynamics is still a good and justified approach. Instead, one could, for example, use some bounded and fast chaotic forcing that acts on the slower dynamics. These two types of forcing can lead to even qualitatively different results: A multi-stable system driven by sufficiently generic unbounded (e.g., Gaussian) noise will almost surely undergo noise-induced tipping in finite time, i.e., if initialized in one attractor of the unforced system, it will transition to another attractor as a result of the noise realization upon waiting for long enough time. In contrast, a system forced by small-amplitude bounded chaos might only tip for certain values of the forced system's parameters, giving rise to a ``chaotic tipping window'' \cite{Ashwin_2024, Roemer2025}. 
Thus, although fast and small fluctuations may sometimes seem negligible, the choice of how to approximate them is essential. 
Despite the differences between stochastic and chaotic forcing, homogenization theory shows that, in the limit of infinite time-scale separation and given some assumptions, chaotic forcing converges to stochastic forcing \cite{Melbourne_2011, Holland_2007, Gottwald_2013}. Close to this limit, one can try to extend results from one case to the other; in this way, we derive the ``chaotic Kramers’ law'' which generalizes the well-known Kramers' law from the case of systems driven by stochastic forcing to systems driven by fast chaotic forcing of large-enough amplitude. 
We check this result in an example system and discuss how to numerically handle the large timescale separation in the chaotically forced system. This example suggests that the chaotic Kramers' law holds far beyond the stochastic limit, even when increments of the chaotically forced system are clearly not normally distributed.
The system that we use as an example is a chaotically forced reduced-order model\cite{Wood_2019} of one of the most important potential climate tipping elements, the Atlantic meridional overturning circulation (AMOC). 

A good understanding of the AMOC is a pressing problem as it strongly affects weather and climate on a global scale. A potential collapse of the AMOC would lead to substantial cooling in Europe and rainfall reductions in tropical African and Asian monsoon systems \cite{Ben-Yami2024}. Observational data suggest declining AMOC strength since the mid-twentieth century after a stable period \cite{Caesar_2021}, but the question of whether the AMOC will undergo an abrupt or large transition (i.e. tipping) due to climate change, and if so on which timescale, remains subject of ongoing research \cite{Boers2021,  Ben-Yami_2023, Ben-Yami2024a, Baker2025AMOC, Oh2025AMOCNoise}.
Studies with highly complex general circulation models (GCMs) can be used to assess the response of the AMOC to anthropogenic greenhouse gas forcing or to artificial hosing, where freshwater is added to the North Atlantic, mimicking meltwater influx from the Arctic \cite{Romanou_2023, Mccarthy_2023}. However, GCMs are computationally expensive and, due to their complexity, sometimes difficult to interpret, which calls for an identification of the most important underlying dynamics using reduced-order models. Due to their computational efficiency, these simpler models can also be used to run larger ensembles or to explore areas in phase space that would be too costly to explore with GCMs. We believe an improved understanding of chaotic forcing, as discussed in this paper, to be a valuable approach to creating useful reduced-order models capturing real-world mechanisms. Our findings indicate that seemingly stochastic transitions could also be observed in purely deterministic GCMs and thus support the use of both deterministic and stochastic AMOC models in future research.

The paper is organized as follows. We give an overview of the theory needed to derive and interpret the chaotic Kramers' law in Section~\ref{sec:Theory} by summarizing results from homogenization theory and discussing the known Kramers' law in the case of stochastically forced bistable systems. Then, we generalize this to the case of chaotically forced systems in Section~\ref{sec:chaotic_kramers}. We consider a reduced-order AMOC model\cite{Wood_2019, Alkhayuon_2019, Chapman_2024} in which we replace the stochastic with chaotic forcing in Section~\ref{sec:3-boxModelAndLorenz} and discuss how the fast chaotic forcing can be connected to stochastic forcing in a numerically stable way. Then, we use this model to test the chaotic Kramers' law numerically in Section~\ref{sec:Kramers'plots} and interpret the results in the context of other studies on the AMOC. We end with a general discussion in Section~\ref{sec:discussion}.

\section{Chaotic limits and asymptotic transition times}
\label{sec:Theory}

\subsection{Homogenization}
In this section, we present one specific result from the theory of \textit{averaging} and \textit{homogenization} of ODEs. The general idea of homogenization is to average out fast time-scales in multiscale systems to get a simpler description of the influence on the slow variables. These techniques can be applied not only for ODEs but also for \textit{partial differential equations} (PDEs) or Markov chains (see \cite{Pavliotis_2008} for an introduction).

We consider a skew-product fast-slow ODE system, where the slow subsystem represents some drift dynamics we want to observe, and the fast subsystem involves chaotic dynamics on a compact attractor. The chaotic subsystem ($y$-dynamics) acts as an additive forcing on the drift ($x$-dynamics), resulting in constant additive noise in the limit of infinitely large timescale separation between the forcing and the forced system. This is a special case of Result 11.1 in \cite{Pavliotis_2008}, which allows the variables $x^{(\epsilon)}$ and $y^{(\epsilon)}$ to be stochastic.

Concretely, we consider a family of fast-slow ODE systems of the form
\begin{subequations} \label{eq:skew_system}
\begin{align}
\dot{x}^{(\epsilon)} &= \epsilon^{-1}f_0(y^{(\epsilon)})+f(x^{(\epsilon)}), &x^{(\epsilon)}(0) = x_0 \in \mathbb{R}^n \label{eq:slow_sub}\\ 
\dot{y}^{(\epsilon)} &= \epsilon^{-2}g(y^{(\epsilon)}), &y^{(\epsilon)}(0) = y_0 \in \mathbb{R}^m \label{eq:fast_sub}
\end{align}
\end{subequations}
where the dynamics of $y$ are supported on a compact attractor $\Lambda \subseteq \mathbb{R}^m$ with ergodic invariant measure $\mu$.

\begin{lemma}
\label{lem:melbourne}
Consider a system of the form \eqref{eq:skew_system}. Let $W_t$ be an $n$-dimensional Wiener Process and $\Sigma$ be an $n \times n$ covariance matrix.\\
We assume that $f_0, f$ and $g$ are locally Lipschitz, $\int_\Lambda f_0(y) d\mu(y) = 0$, $x^Tf(x) \leq M(1 + \lVert x\rVert)$ for some $M > 0$ and that $(y, f_0, \Sigma)$ satisfies the weak invariance principle \ref{def:WIP}. 

Then $x^{(\epsilon)} \rightarrow_w X$ in $C([0, T]), \mathbb{R}^n)$ for $\epsilon \rightarrow 0$ for any $T > 0$, where $X$ is the unique solution to the SDE 
\begin{align}
dX = f(X) dt + \sqrt{\Sigma}dW_t
\end{align}
\end{lemma}
This result can be proven using the approach by \cite{Melbourne_2011}, from which we adopted the result and adjusted it slightly.
Note that the diffusion matrix $\sqrt{\Sigma}$ can be computed from time correlations of the fast dynamics via the Green-Kubo formula or directly estimated from the limiting process, as discussed in Section \ref{sec:Green-Kubo_vs_WIPestimation}.

\begin{definition}[Weak invariance principle]\cite{Melbourne_2011} \label{def:WIP}
We say the triple $(y, f_0, \sqrt{\Sigma})$ with the properties given in Lemma \ref{lem:melbourne} satisfies the weak invariance principle (WIP) if for all $T > 0$:
\begin{align}
\label{eq:WIP}
k^{-\frac{1}{2}} \int_0^{kt}f_0(y^{(1)}(\tau)) d\tau \rightarrow_w \sqrt{\Sigma} W_t
\end{align}
for $k \rightarrow \infty$ in $C([0, T], \mathbb{R}^n)$.
\end{definition}
By $\rightarrow_w$, we denote weak convergence of measures (distributions), where we interpret $y^{(1)}$ as a random variable by drawing the initial condition according to the invariant probability measure $\mu$.
One can think of this weak invariance principle as a kind of a functional Central Limit Theorem for $f_0(y)$, in which the slow variable converges in distribution as a 
stochastic process to the solution of an effective SDE. Proving the validity of the WIP for a specific system is highly non-trivial and often requires powerful concepts from dynamical systems theory.
Despite that, according to \cite{Melbourne_2011}, this assumption is true for a large class of flows, including Lorenz attractors. In this specific case, an even stronger version of \ref{def:WIP} is true, as $\mu$-almost sure convergence for geometric Lorenz flows can be proven\cite{Holland_2007}.

\subsection{Kramers' law and large deviation theory}

We briefly summarize different formulations of the classical Kramers' law for stochastically forced systems and state an underlying result from Freidlin and Wentzell's theory of large deviations. In the next section, we then generalize it to the case of time-scale separated chaotic forcing, building on a sketch of a derivation of the classical Kramers' law, given in Appendix \ref{app:LargeDevToKramers}.

Kramers' law gives an asymptotical result on expected transition times in multi-stable potential systems with white noise and was originally formulated in the modeling of chemical reactions, see \cite{Arrhenius_1889, Kramers_1940}. A broad overview of Kramers' law and its derivation from large deviations theory is given in \cite{Berglund_2013}. We summarize this approach in \Cref{sec:ldp_to_kramers}. 
Consider a stochastic differential equation of the form
\begin{align}
\label{eq:Kramer_SDE}
dX = -\nabla V(X)dt+\sqrt{2\delta}dW_t\text{,}
\end{align}
where $V: \mathbb{R}^n \rightarrow \mathbb{R}$ is a smooth potential with two wells at the local minima $\bar x_1$ and $\bar x_2$, and $W_t$ is an $n$-dimensional Wiener Process. We assume enough dissipativity such that equation \eqref{eq:Kramer_SDE} has a unique global-in-time strong solution $(X_t)_{t \geq 0}$.
Let $\tau_{1, 2}$ be the first hitting time of $X$ of a ball of radius $\delta$ around $\bar x_2$, when starting the process in $\bar x_1$.
We define the \textit{ communication height} as the minimal height of the potential crossing from $\bar x_1$ to $\bar x_2$ over a continuous path $\gamma$:
\begin{align}
H(\bar x_1, \bar x_2) = \inf_{\gamma: \bar x_1 \rightarrow \bar x_2}\left(\sup_{z \in \gamma} V(z)\right).
\end{align}
This minimal height of the potential $V(z)$ is a measure of the difficulty of crossing from $\bar x_1$ to a ball of radius $\delta$ around $\bar x_2$.
For a generic $V$, there is a unique point $x_s$ that attains the communication height, i.e. $H(\bar x_1, \bar x_2) = V(x_s)$. This $x_s$ is called the relevant saddle and is a critical point of $V$ where $\nabla^2 V(x_s)$ has $n-1$ positive and one negative eigenvalue.

\begin{theorem}[Kramers' law] \cite{Berglund_2013}
\label{thm:KramersLaw}
For small $\delta$ one has
\begin{align}
\mathbb{E}_{\bar x_1}[\tau_{1, 2}] = C e^{\frac{1}{\delta}[V(x_s) - V(\bar x_1)]} \left[ 1+\mathcal{O}(\delta^{\frac{1}{2}} \lvert \log{\delta} \rvert^{\frac{3}{2}}) \right]\text{,}
\end{align}
where $C$ is a constant depending on $V$.
\end{theorem}
The subscript in $\mathbb{E}_{\bar x_1}$ indicates that the process starts in $\bar x_1$.
Using the stationarity of the \textit{Gibbs-measure} $\rho(x) = \frac{1}{Z_{(\delta)}}e^{-\frac{V(x)}{\delta}}$ for \eqref{eq:Kramer_SDE} if $Z_{(\delta)} = \int_{\mathbb{R}^n} e^{-\frac{V(x)}{\delta}} dx < \infty$, one can determine $C$ explicitly\cite{Bouchet_2016} to be $\frac{2\pi}{\lvert \lambda_1 \rvert} \sqrt{\frac{\lvert \det{\nabla^2V(x_s)}\rvert}{\det{\nabla^2 V(\bar x_1)}}}$, where $\lambda_1$ is the unique negative eigenvalue of $\nabla^2 V(x_s)$. 
We now turn to the more general case of irreversible and anisotropic systems, i.e., systems for which the drift cannot be expressed as a gradient. These are of the general form $dX = F(X)dt + \sqrt{2\delta} A dW$.
We introduce a {\em quasipotential}, which is a useful generalization of the potential $V$ for irreversible systems. To prepare its Definition \ref{def:QP}, we first introduce a few objects.

Consider a Polish space $E$ with Borel $\sigma$-algebra $\mathcal{E}$. We define a {\em rate function} $I : E \rightarrow [0, \infty]$ as a lower semicontinuous function, i.e. $N_\alpha = \{x\in E | I(x) \leq \alpha\}$ is closed for all $\alpha \in \mathbb{R}\cup\{\infty\}$. If, in addition, the sublevel sets $N_\alpha$ are compact, $I$ is called a good rate function \cite{dembo_large_2010}. 
In the following, we consider the {\em Freidlin-Wentzell action} \cite{Freidlin_wentzell_2012} $I:C([0, T], \mathbb{R}^n) \rightarrow \mathbb{R} \cup \{ \infty\}$ as the good rate function. 
It characterizes the ``difficulty'' of a given path $\varphi$. Typically, one considers a path that is initialized in an attractor $x_0$.
$I$ is defined as:
\begin{align}
I(\varphi) = 
\begin{cases}\label{eq:FW-action}
	\frac{1}{2}\int_0^T \lVert \dot{\varphi}(t)-F(\varphi(t))\rVert_A^2 dt \text{ if } \varphi \in \mathcal H_1 \\
    \infty  \text{ else} 
\end{cases}
\end{align}
where $\lVert x \rVert_A^2 := x^\mathsf{T}(AA^\mathsf{T})^{-1}x$ and $\mathcal H_1$ denotes the Cameron-Martin space 
\begin{align*}
\mathcal{H}_1 = \{\varphi :[0, T] \rightarrow \mathbb{R}^n \, \text{absolutely continuous} \\ \text{s.t. } \varphi(0) = x_0, \dot \varphi \in L^2\}
\end{align*}
with norm $||\varphi||_{\mathcal{H}_1}=\left(\int_0^T||\dot\varphi (s)||^2ds\right)^{\frac{1}{2}}$.
Note that it contains the same functions as the Sobolov space $H^1([0,T], \mathbb{R}^n)$ with the additional boundary condition $h(0)=x_0$.

Now, we can define the quasipotential via a double-minimization of the FW action over both the possible paths from the reference attractor towards the point of interest and over the travel time along these paths.

\begin{definition}[Quasipotential] \label{def:QP}\cite{Berglund_2013}
Assume $D \subset \mathbb{R}^n$ s.t. $D$ is contained in the basin of attraction of the asymptotically stable equilibrium $\bar x_1$ of the unperturbed system $\dot x = F(X)$ and let $I$ be given by Equation \eqref{eq:FW-action}. For $z\in D$ we define
\begin{align*}
\bar V(\bar x_1, z)= \inf_{T > 0} \inf_{\varphi} \{I(\varphi):\varphi \in (C[0, T], \mathbb{R}^n): \varphi(0) = \bar x_1, \\\varphi(T) = z\} \text{.}
\end{align*}
$\bar V$ is called a quasipotential.
\end{definition}
In the case of a gradient system, the quasipotential is just an affine transformation of the potential of the system $V(x)$.

\begin{theorem}[The irreversible Kramers formula]\cite{Bouchet_2016} \label{thm:irrev_Kramers}
Consider a diffusion process of the form $dX = F(X)dt + \sqrt{2\delta} A dW$ with $F$ Lipschitz and bounded, $A \in \mathbb{R}^{n \times n}$ non-degenerate, where the unperturbed system $\dot x = F(x)$ has two asymptotically stable equilibria $\bar x_1$ and $\bar x_2$ and a relevant saddle $x_s$. Under a few more assumptions (see \cite{Bouchet_2016}, subsection 5.1.3.) we find that
\begin{equation}
\label{eq:irrev_kramers}
\lim_{\delta \rightarrow 0}2\delta \log \mathbb{E}_{\bar x_1}[\tau_{1, 2}] =\bar V(\bar x_1, x_s)\text{.}
\end{equation}

\end{theorem}

Equation \eqref{eq:irrev_kramers} is sometimes written as
\begin{align}
\label{eq:irrev_kramersLogEq}
\mathbb{E}_{\bar x_1}[\tau_{1, 2}] \simeq e^{\frac{\bar V(\bar x_1, x_s)}{2\delta}}\text{,}
\end{align}
where the symbol ``$\simeq$'' means {\em logarithmic equivalence} as defined in \cite{heymann_geometric_2008}: two functions $f(x)$ and $g(x)$ are logarithmically equivalent if and only if $\lim_{\delta \rightarrow 0}\frac{\log(f(\delta))}{\log(g(\delta))} = 1$.

In \cite{bouchet_generalisation_2016}, a prefactor $\tilde C$ for Equation \eqref{eq:irrev_kramersLogEq} is derived explicitly, s.t. they write:
\begin{align}
\label{eq:irrev_kramersLogEqCtilde}
\mathbb{E}_{\bar x_1}[\tau_{1, 2}] = \lim_{\delta\rightarrow0^+}\tilde C e^{\frac{\bar V(\bar x_1, x_s)}{2\delta}}\text{.}
\end{align}

Equation \eqref{eq:irrev_kramers} can be established using  Freidlin-Wentzell theory; some details are given in the Appendix \ref{app:LargeDevToKramers}.

\subsubsection{Generalization: chaotic Kramers' law}
\label{sec:chaotic_kramers}
We now transfer Kramers' law to the chaotic case. We first state the result in Lemma \ref{lem:chaotic_kramers} and then derive it by first deriving Theorem \ref{thm:chaotic_ldp}, a large deviation principle (LDP) in the chaotic case for a specific set $\Gamma$, and then use this LDP to deduce the chaotic Kramers' law \ref{lem:chaotic_kramers}. This is analogous to the stochastic case, which is summarized in Appendix \ref{app:LargeDevToKramers}.

Let $\Gamma \subset \mathcal{H}_1$ be a regular set of paths with respect to $I$, i.e., the minimization of $I$ over the interior $\Gamma^\circ$ equals the minimization over the closure $\bar \Gamma$, so $\inf_{\varphi \in \Gamma^\circ}I(\varphi) = \inf_{x \in \bar \Gamma}I(x)$.

Denote by $x^{(\epsilon)}_{(\delta)}$ the solution of the ODE \eqref{eq:slow_sub} with $f_0 = \sqrt{2 \delta}\pi_1$, and by $X_{(\delta)}$ the solution of the limit SDE. We assume the unforced dynamics of $x$ to have stable equilibria at $\bar x_1$ and $\bar x_2$, and we would like to investigate tipping from $\bar x_1$ to $\bar x_2$ in response to both stochastic and chaotic forcing. Assume all solutions are defined up to some finite time $T \geq 0$.

\begin{lemma}[chaotic Kramers' law] \label{lem:chaotic_kramers}
    Let $\tau_{1, 2}^{(\epsilon)}$ be the first hitting time of the stochastic process $x^{(\epsilon)}_{(\delta)}$, initialized in $\bar{x}_1$ of a small neighborhood of the equilibrium $\bar x_2$. 
    For $\epsilon \rightarrow 0$, $\epsilon \ll \delta$,
    \begin{align}
        \liminf_{\delta \rightarrow 0} 2\delta \log \mathbb{E}_{\bar{x}_1}[\tau_{1, 2}^{(\epsilon)}] = \bar V(\bar x_1, x_s) \text{.}
    \end{align}
\end{lemma}
Note that $x^{(\epsilon)}_{(\delta)}$ is a stochastic process in the sense described after Definition \ref{def:WIP}. In the following, we derive Lemma \ref{lem:chaotic_kramers}.

({\em Proof outline}) We consider paths from $\bar x_1$ to the critical saddle $x_s$:
\begin{equation}
\begin{aligned}
\Gamma = \{x \in C([0, T], \mathbb{R}^n) \text{ s.t. } \exists S > 0 \text{ with }\\ x(0) = \bar x_1, x(S) \in B_{\frac{1}{k}}(x_s)\} \text{,}
\end{aligned} 
\end{equation}
where $B_{\frac{1}{k}}(x_s)$ is an open ball of radius $k^{-1}$ around $x_s$ with $k \in  \mathbb{N}$  large, which we consider to get positive hitting probabilities and thus, finite hitting times. Assume that $T$ is chosen to be sufficiently large for the instanton $\psi$ from $\bar x_1 $ to $x_s$ (as introduced after Theorem \eqref{thm:transformed_FWT}) to connect the two balls within $[0, T]$. In the following, we will often suppress the dependence on $T$, $k$, and the initial condition $\bar x_1$. 

We now derive Theorem \ref{thm:chaotic_ldp}, an LDP for the set of paths $\Gamma$ of the chaotic system in the limit $\varepsilon\rightarrow0$. To do so, we require the following result.

\begin{lemma} \label{lem:convergence_gamma}
For $\Gamma$, $x^{(\epsilon)}_{(\delta)}$ and $X_{(\delta)}$ as above
\begin{align}
\mathbb{P}[x^{(\epsilon)}_{(\delta)} \in \Gamma] \rightarrow \mathbb{P}[X_{(\delta)} \in \Gamma]
\end{align}
for $\epsilon \rightarrow 0$ for every fixed $\delta > 0$.
\end{lemma}
\begin{proof}
First, notice that both processes have the same initial condition, i.e. 
\begin{align*}
\mathbb{P}[x^{(\epsilon)}_{(\delta)}(0) = \bar x_1] = 1 =  \mathbb{P}[X_{(\delta)}(0) = \bar x_1]
\end{align*}  
by definition. 
Define the stopping times 
\begin{align*}
\bar S_{(\delta)} = \inf \{S \in [0, T] : X_{(\delta)}(S) \in B_{\frac{1}{k}}(x_s)\}
\end{align*} 
and
\begin{align*}
\bar S_{(\delta)}^{(\epsilon)} = \inf \{S \in [0, T] : x^{(\epsilon)}_{(\delta)}(S) \in B_{\frac{1}{k}}(x_s)\}
\end{align*} with the convention $\inf(\emptyset) = \infty$.

Lemma \ref{lem:melbourne}, which relies on the WIP \ref{def:WIP}, showed the convergence of the distributions $x^{(\epsilon)}_{(\delta)} \rightarrow_w X_{(\delta)}$.
The functional $h:C([0, T], \mathbb{R}^n) \rightarrow [0, \infty], h(x) = \inf \{S \in [0, T] : x(S) \in B_{\frac{1}{k}}(x_s)\}$ is almost surely continuous for the limit process $X_{(\delta)}$, i.e., the set of discontinuities $D_h$ of $h(X_{(\delta)})$ has measure zero: $\mathbb P[X_{(\delta)} \in D_h]=0$. This is the case as the functional $h$ is discontinuous exactly at paths that: (a) stay arbitrarily close to $\partial B_{\frac{1}{k}}(x_s)$ without entering $B_{\frac{1}{k}}(x_s)$, or (b) hit the boundary $\partial B_{\frac{1}{k}}(x_s)$ at a first time $\tau$ but fail to enter the interior within an arbitrarily short interval afterward.
Due to the positivity and smoothness of transition densities and the strong Markov property of $X_{(\delta)}$, we know that both (a) and (b) have probability zero \cite{Billingsley1999}.

Thus, we can apply the continuous mapping theorem \cite{vanderVaart1998} to ensure that the distributions of $\bar S_{(\delta)}^{(\epsilon)}$ converge to those of $\bar S_{(\delta)}$.
Then, given that the initial conditions are concentrated in $\bar x_1$,
\begin{align*}
\mathbb P[x^{(\epsilon)}_{(\delta)} \in \Gamma] = \mathbb P[\bar S_{(\delta)}^{(\epsilon)} < \infty] \rightarrow \mathbb P[\bar S_{(\delta)} < \infty] = \mathbb P[X_{(\delta)} \in \Gamma]\text{.}
\end{align*}
\end{proof} 

Thus for any $\epsilon_1 > 0$, there is an $\tilde \epsilon=\tilde \epsilon(\epsilon_1, \delta, \Gamma)$ s.t. 
\begin{align} \label{eq:probConvEps}
\lvert \mathbb{P}[X_{(\delta)} \in \Gamma] - \mathbb{P}[x^{(\epsilon)}_{(\delta)} \in \Gamma]\rvert < \epsilon_1
\end{align} 
for all $\epsilon \leq \tilde \epsilon$. By the properties of the Wiener process and the definition of $\Gamma$ via hitting a ball
of non-zero Lebesgue measure around $x_s$, we can ensure $\mathbb{P}[X_{(\delta)} \in \Gamma]>0$.  

Equation (\ref{eq:probConvEps}) implies the two inequalities
\begin{equation}
\begin{aligned}\label{eq:twoIneqalities}
\mathbb{P}[x^{(\epsilon)}_{(\delta)}\in \Gamma] - \epsilon_1 \leq{\mathbb{P}[X_\delta \in \Gamma]} &\leq \mathbb{P}[x^{(\epsilon)}_{(\delta)}\in \Gamma] + \epsilon_1.
\end{aligned}
\end{equation}
In the following, we choose $\epsilon_1$ small, s.t. $\epsilon_1 \ll \frac{1}{2} \mathbb{P}[X_{(\delta)}\in \Gamma]< \mathbb{P}[x^{(\epsilon)}_{(\delta)}\in \Gamma]$. This implies $\chi:=\epsilon_1 \mathbb{P}[x^{(\epsilon)}_{(\delta)}\in \Gamma]^{-1} \ll 1$.
Taking the logarithm of (\ref{eq:twoIneqalities}) and applying the Taylor expansion $\log(y+\epsilon_1)=\log(y)+\sum_{i = 1}^\infty \frac{(-1)^{i+1}}{i}(\frac{\epsilon_1}{y})^i$ leads to the two inequalities
\begin{equation}
\begin{aligned} \label{eq:InequpperlogTaylor}
\log{\mathbb{P}[X_\delta \in \Gamma]} &\leq \log (\mathbb{P}[x^{(\epsilon)}_{(\delta)}\in \Gamma] + \epsilon_1)\\ &= \log (\mathbb{P}[x^{(\epsilon)}_{(\delta)}\in \Gamma])+ \chi +\mathcal{O}(\chi^2)
\end{aligned}
\end{equation}
and 
\begin{equation}
\begin{aligned} \label{eq:IneqlowerlogTaylor}
\log{\mathbb{P}[X_\delta \in \Gamma]} &\geq \log (\mathbb{P}[x^{(\epsilon)}_{(\delta)}\in \Gamma] - \epsilon_1)\\ &= \log (\mathbb{P}[x^{(\epsilon)}_{(\delta)}\in \Gamma])-\chi +\mathcal{O}(\chi^2).
\end{aligned}
\end{equation}
The higher-order terms of the logarithm expansion contribute less than the first-order term and are negative in total: if $y, \epsilon_1>0$ and $y > \epsilon_1$ then $-\frac{\epsilon_1}{y} < \sum_{i = 2}^\infty \frac{(-1)^{i+1}}{i}(\frac{\epsilon_1}{y})^i < 0$, and we thus find
\begin{align}\label{eq:logHigherOrderTerms}
\log(y)-2\frac{\epsilon_1}{y} < \log(y-\epsilon_1) < \log(y+\epsilon_1) < \log(y)+\frac{\epsilon_1}{y}\text{.}
\end{align}

We assume $\Gamma$ to be regular with respect to $I$ as described before Theorem \ref{thm:transformed_FWT}. Using the LDP (\ref{eq:LDP}) and the inequalities \ref{eq:InequpperlogTaylor}, \ref{eq:IneqlowerlogTaylor}, and \ref{eq:logHigherOrderTerms}, we find for all $\epsilon \leq \tilde \epsilon$ 
\begin{equation}
\begin{aligned}
- \inf_{x \in \bar\Gamma}I(x) &\leq
\lim_{\delta \rightarrow 0} 2\delta \log (\mathbb{P}[X_{(\delta)}\in \Gamma]) \\
& \leq \lim_{\delta \rightarrow 0} 2\delta \left(\log (\mathbb{P}[x^{(\epsilon)}_{(\delta)}\in \Gamma])+ \chi \right) \\
&= \lim_{\delta \rightarrow 0} 2\delta (\log (\mathbb{P}[x^{(\epsilon)}_{(\delta)}\in \Gamma]))+ \lim_{\delta \rightarrow 0}2\delta \chi  \\
&=\lim_{\delta \rightarrow 0}2 \delta (\log (\mathbb{P}[x^{(\epsilon)}_{(\delta)}\in \Gamma]))\text{,}
\end{aligned}
\end{equation}
where $\lim_{\delta \rightarrow 0}\delta \chi^n = 0$ for all $n \in \mathbb{N}$, as $0 < \chi < 1$. 
Similarly 
\begin{equation}
\begin{aligned}
- \inf_{x \in \bar\Gamma}I(x) &\geq
\lim_{\delta \rightarrow 0} 2\delta \log (\mathbb{P}[X_{(\delta)}\in \Gamma]) \\
& \geq \lim_{\delta \rightarrow 0} 2\delta \left(\log (\mathbb{P}[x^{(\epsilon)}_{(\delta)}\in \Gamma])- 2\chi \right) \\
& = \lim_{\delta \rightarrow 0} 2\delta (\log (\mathbb{P}[x^{(\epsilon)}_{(\delta)}\in \Gamma])) \text{.}
\end{aligned}
\end{equation}
This proves $\mathbb{P}[x^{(\epsilon)}_{(\delta)}\in \Gamma] \simeq e^{\frac{I(\varphi)}{2\delta}}$ in the limit $\epsilon\rightarrow 0$.\\
If we don't evaluate the limits of the terms containing $\chi$ directly, we can find first-order error bounds on the change of the subexponential prefactor.

\begin{theorem}[Chaotic LDP] \label{thm:chaotic_ldp}
Given that $\epsilon \rightarrow 0$ sufficiently fast compared to $\delta\rightarrow0$, 
\begin{equation}
\begin{aligned}
&\ \ \ \, \,- \inf_{x \in \Gamma}I(x) - \lim_{\delta \rightarrow 0 }2\delta \chi \leq \lim_{\delta \rightarrow 0} 2\delta \log (\mathbb{P}[x^{(\epsilon)}_{(\delta)}\in \Gamma]) \\ 
&\leq
- \inf_{x \in  \Gamma}I(x) + \lim_{\delta \rightarrow 0 }4\delta \chi
\end{aligned}
\end{equation}
\end{theorem}

As $0 < \chi < 1$, we can derive from Theorem 
\ref{thm:chaotic_ldp}
\begin{equation}
\begin{aligned}
&\ \ \ \, \,- \inf_{x \in \Gamma}I(x) - \lim_{\delta \rightarrow 0 }2\delta \leq \lim_{\delta \rightarrow 0} 2\delta \log (\mathbb{P}[x^{(\epsilon)}_{(\delta)}\in \Gamma])
\\ &
\leq - \inf_{x \in  \Gamma}I(x) + \lim_{\delta \rightarrow 0 }4\delta,
\end{aligned}
\end{equation}
and thus, we can write 
\begin{equation}
\begin{aligned}\label{eq:chaoticLDP}
\mathbb{P}[x^{(\epsilon)}_{(\delta)}\in \Gamma] \simeq  e^{-\frac{1}{2\delta}\inf_{x \in  \Gamma}I(x)} \simeq \mathbb{P}[X_{(\delta)}\in \Gamma].
\end{aligned}
\end{equation}

This derivation is only valid if $\chi < 1$, which can be guaranteed if $\epsilon_1 < \frac{1}{2}P[X_{(\delta)}\in \Gamma] \simeq e^{-\frac{1}{2\delta}\inf_\Gamma I} \rightarrow 0$ as $\delta \rightarrow 0$. Although we do not know the exact relationship between $\epsilon_1$ and $\tilde\epsilon$, this suggests that $\epsilon$ has to approach $0$ at least exponentially fast compared to $\delta$.

The derivation of the exact sub-exponential prefactor to write Equation \eqref{eq:chaoticLDP} not just as logarithmic equivalence, but as equality, is left for future work. However, the terms containing $\chi$ in Theorem \ref{thm:chaotic_ldp} suggest an additional prefactor $\kappa\in [e^{-1}, e^2]$.
\vspace{\baselineskip}

Analogous to the procedure in the SDE case explained after Equation (\ref{eq:exit_time_D}), we can finally derive the chaotic Kramers' law, Lemma \ref{lem:chaotic_kramers}, from Equation (\ref{eq:chaoticLDP}).
In the sketch of the proof of Equation \eqref{eq:exit_time_D}, we defined $p = e^{-\frac{\bar V(\bar x_1, x_s)}{2\delta}}$. Now, let $\tilde p = \kappa p \in [0, 1]$ for sufficiently small $\delta$. In the SDE case, we then used the Markov property to divide the process into intervals of time $T_1$ for deriving bounds for the expected exit time from $D$. By contrast, the process $x^{(\epsilon)}_{(\delta)}$ with fixed $\epsilon$ and $\delta$, is not Markovian, as we have to view it as a solution process of the inhomogeneous random ODE 
\begin{align}
    \dot x^{(\epsilon)}_{(\delta)} = \epsilon^{-1}f_0\left(y^{(\epsilon)}_0+\int_0^t \epsilon^{-2} g(y^{(\epsilon)}(s)) ds\right) + f(x^{(\epsilon)}_{(\delta)})
\end{align} 
converging to the SDE solution. The dependence on $\delta$ will be reflected in $f_0$.

Nevertheless, we conjecture that under moderate assumptions on the chaotic subsystem and large enough time scale separation, an approximate Markov property is satisfied when observing the forced subsystem only after sufficient evolution times of the chaotic forcing. Thus, restarting the process at multiples of $T_1$ to calculate the expected exit time from the geometric distribution is still justified in the chaotic case, similar to the SDE case.

Note that using $\tilde p$ yields a possibly larger upper bound $\frac{T_1}{\kappa}\exp\left(\frac{\bar V(\bar x_1, x_s)}{2\delta}\right)$, but doesn't change the limit argument. Thus, we find Lemma \ref{lem:chaotic_kramers}. (End of proof outline)  

Although the derivation presented here can't give a statement on the subexponential prefactor, it suggests that the prefactor changes within the range of $\kappa^{-1} \in [e^{-2}, e^1]$. Confirming this hypothesis is left for future work. 

Note that for a fixed $\epsilon > 0$, estimating $\mathbb E[\tau_{1, 2}^{(\epsilon)}]$ would require concrete knowledge about the invariant measure on the compact attractor. As $f_0(\Lambda)$ is compact, for fixed $\epsilon$, the perturbation is bounded, while the increments of the Wiener Process are unbounded. This means there might be a $T_0 \in [0, \infty]$ s.t. the probability of tipping in the chaotic system is $0$, for some choices of $\epsilon$ and $\delta$, while it is non-zero in the SDE. 
Thus, the chaotic Kramers' law cannot hold arbitrarily far from the stochastic limit $\varepsilon\to 0$ and will break down at some point. However, in the example we consider in the following section, we find the chaotic Kramers' law to hold surprisingly far from the stochastic limit, i.e., for $\epsilon \leq 4$.

\section{3-box AMOC model and Lorenz-63 model}
\label{sec:3-boxModelAndLorenz}

The AMOC box model we employ for our study was proposed by \cite{Wood_2019} and is widely utilized in current research as it qualitatively captures the AMOC dynamics of more complex models well. Here, we consider two versions of it: the previously existing version with stochastic forcing\cite{Chapman_2024} and one where we exchanged the stochastic for chaotic forcing. The chaotic forcing that we use here is a toy model representing the forcing that acts on the real-world AMOC by many different mechanisms like atmospheric dynamics, other ocean currents, etc. 
The AMOC strength and salinity concentrations are modeled by interactions between 5 different ocean boxes, the \textit{Northern Atlantic Deep Water}box $N$, the \textit{Atlantic Thermocline} box $T$, the \textit{Southern Ocean near-surface water} box $S$, the \textit{Indo-Pacific Thermocline} box $IP$, and the deep-sea southward propagation of the \textit{North Atlantic Deep Water} box $B$. A visualization of the model structure is given in \cite{Wood_2019}.

By assuming conservation of the total salt content and identifying relevant variables for studying tipping mechanisms, the model can be reduced\cite{Alkhayuon_2019}, leading to the (rescaled) 3-box model described by the following equations:
\begin{equation}
\begin{aligned}
\frac{d \tilde S_N}{dt} 
&= C_N \left[ q \times (\tilde S_T - \tilde S_N) + K_N  (\tilde S_T - \tilde S_N) \right. \\
& \left. - 100  (F_N + A_N H)  S_0 \right] =: f_1^+\\
\frac{d \tilde S_T}{dt} 
&= C_T \left[ q \times (\gamma \tilde{S}_S + (1 - \gamma)\tilde{S}_{IP} - \tilde{S}_T) + K_S(\tilde{S}_S - \tilde{S}_T) \right. \\
& \left. + K_N(\tilde{S}_N - \tilde{S}_T) - 100 (F_T + A_T H)  S_0 \right]=: f_2^+
\end{aligned}
\end{equation}

if $q \geq 0$,
and
\begin{equation}
\begin{aligned}
\frac{d \tilde S_N}{dt} &= C_N \left[ -q \times(\tilde S_B - \tilde S_N) + K_N  (\tilde S_T - \tilde S_N) \right. \\
& \left.  - 100  (F_N + A_N H)  S_0 \right] =: f_1^-\\
\frac{d \tilde S_T}{dt} &= C_T \left[ -q \times (\tilde S_N - \tilde S_T) + K_S(\tilde{S}_S - \tilde{S}_T) \right.\\
& \left. + K_N(\tilde{S}_N - \tilde{S}_T) - 100 (F_T + A_T H)  S_0 \right] =: f_2^-
\end{aligned}
\end{equation}
 if $q < 0$,
where
\begin{align*}
q(\tilde S_N) &= \frac{{\lambda \left[ \alpha \left( T_S - T_0 \right) + \frac{\beta}{100} \left( \tilde S_N - \tilde S_S \right) \right]}}{{1 + \lambda  \alpha  \mu}}\text{.} 
\end{align*} and 

\begin{align*}
    \tilde S_{IP}(\tilde S_N, \tilde S_T) &= \frac{1}{V_{IP}}\left[100(C - S_0(V_B+V_N+V_T+V_IP+V_S)) \right. \\
    &- \left. (V_N \tilde S_N + V_T \tilde S_T+ V_S \tilde S_S + V_B \tilde S_B) \right]
\end{align*}

This model is given in the rescaled version where time is given in years and $35+10\tilde S_i$ has unit ppt (parts per thousand)\cite{Alkhayuon_2019}. The model parameters can be found in the Appendix \ref{sec:Appendix}. We included an optional hosing $H$ (releasing freshwater into the Northern Atlantic, which may be interpreted as influx from melting of the Greenland Ice Sheet and Sea Ice, as well as increased precipitation and river runoff). It is known that the model exhibits a hysteresis in the bifurcation parameter $H$, and one can observe rate-dependent tipping behavior upon considering $H$ as a time-dependent forcing \cite{Alkhayuon_2019}. Some sample time series are shown in \Cref{fig:hosing_comp}. 

Defining $f_i = f_i^+ \mathbbm{1}_{q \geq 0} + f_i^- \mathbbm{1}_{q<0}$, we will denote this concisely as two-dimensional homogeneous ODE $\frac{dx}{dt} = f(x)$ with $x=(\tilde S_N, \tilde S_T)$.

The model's parameters can be adjusted to match the behavior of the AMOC in different GCMs\cite{Wood_2019, Chapman_2024_2} (see Appendix \ref{sec:Appendix}). Previous work\cite{Chapman_2024_2} also studied the influence of noise on the AMOC tipping behavior, estimated noise amplitudes for the three-box model, and analyzed the SDE dynamics $dX = f(X)dt+AdW_t$. In this work, we will focus on the new HadGEM3-MM model parameter calibration, yielding a hysteresis in the hosing parameter $H$, and the corresponding HadGEM3-MM noise coefficient matrix
\begin{align}
A_{MM} = \left(
\begin{matrix}
0.1263 & 0 \\
-0.0869 & 0.1088 
\end{matrix}
\right) \text{.}
\end{align}

Notice that we scale the model \cite{Alkhayuon_2019} for numerical reasons, so the scaling of $A_{MM}$ needs to be adopted to match the model's units. The noise matrix \cite{Chapman_2024_2} results in rather small noise compared to data seen in other studies \cite{Boers2021, Ben-Yami2024a, Ben-Yami_2023, VanWesten2024}. As we are only interested in a qualitative analysis here, we scale it by a factor of $10$ to see tipping in acceptable computational time.

The theory of Large Deviations and the classical and chaotic Kramers' law formally require boundedness and Lipschitz continuity, which are not given in this 3-box model. Nevertheless, we can hope to apply these by exploiting the dissipativity in the system. This indicates that the SDE solution exhibits a tight invariant measure and that the relevant dynamics can be restricted to a compact set. Some useful results can be found in \cite{Khasminskii_2012}, chapter 3.

\begin{figure}
\includegraphics[width=\linewidth]{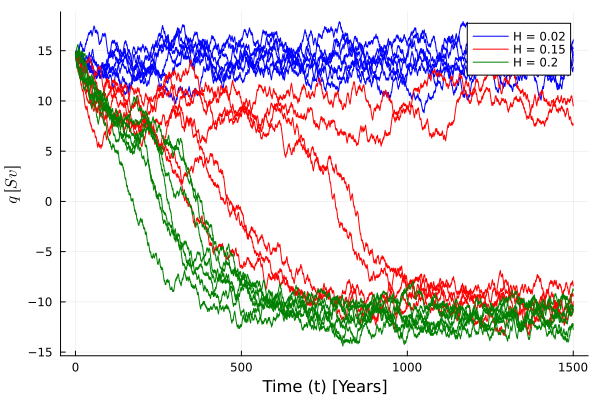}
\caption{Sample timeseries of the chaotically forced 3-box AMOC model with $\epsilon = 1$ and $\delta = 0.5$ for different hosing values. For all three values of $H$, the deterministic unforced system is bistable.}
\label{fig:hosing_comp}
\end{figure}

We are interested in a skew-product ODE system of the form of Equation \eqref{eq:skew_system} that satisfies the assumptions of \Cref{lem:melbourne} s.t. the solution of the slow component converges to the solution of 
\begin{align}\label{eq:SDE}
    dX &= f(X)dt+\sqrt{2\delta}AdW_t.
\end{align} We introduce an additional noise scaling $\sqrt{2\delta}$ in order to apply Kramers Formula later on.

For the chaotic component \eqref{eq:fast_sub}, we choose two independent standard Lorenz-63 systems.
It might be possible to find two observables of a single Lorenz system that, being used in \Cref{lem:melbourne}, yield two independent Wiener Processes. However, this can be achieved more easily by considering a second Lorenz-63 system with a different initial condition. Consider $\dot y^{(\epsilon)} = \epsilon^{-2}g(y^{(\epsilon)})$, such that $y^{(\epsilon)} \in \mathbb{R}^6$.

We define $\pi_i: \mathbb{R}^l \rightarrow \mathbb{R}$ to be the projection onto the i'th coordinate w.r.t. the standard basis. By \Cref{lem:melbourne} with $f \equiv 0$ and $f_0 = \pi_1$, $\epsilon^{-1}\pi_1g(y^{(\epsilon)})$ converges to a scaled Wiener Process $\sigma_L W_t$.

We define 
\begin{align}
f_0 = \frac{\sqrt{2\delta}}{\sigma_L}A
\left(
\begin{matrix}
\pi_1\\
\pi_4
\end{matrix}
\right) 
\end{align}

and claim that the solution of $\dot x^{(\epsilon)} = \epsilon^{-1}f_0(y^{(\epsilon)})+f(x^{(\epsilon)})$ converges to the solution of $dX = f(X) dt + \sqrt{2\delta}AdW_t$ due to the linearity of the WIP \ref{def:WIP} and weak convergence.
This relies on the assumptions that different initial conditions on the Lorenz attractor yield independent noise realizations in the limit and that the coupling results in the anticipated covariance. Looking at the chaotic properties of the Lorenz attractor and the linearity of the WIP and weak convergence, these assumptions seem reasonable. A rigorous proof is left for future work. 

\subsection{Green-Kubo-Formula approximation vs. WIP estimation}
\label{sec:Green-Kubo_vs_WIPestimation}
We aim to calculate $\sigma_L$ in order to be able to simulate the chaotic skew product ODE numerically. We can express $\sigma_L$ explicitly using the Green-Kubo formula

\begin{align}
\frac{1}{2}\sigma_L^2 = \int_0^\infty \lim_{S \rightarrow \infty}  \int_0^S y_1(t)y_1(t+s) ds \text{ } dt \text{,}
\end{align}
where $y_1 = \pi_1 \circ y^{(1)}$. See \cite{Pavliotis_2008}, 11.7.2

Via ergodicity \cite{araujo_singular-hyperbolic_2008, ARAUJO2023373}, this is equivalent to 

\begin{align}\label{eq:GreenKubo}
\frac{1}{2}\sigma_L^2 = \int_0^\infty \int_\Lambda y_1^{(1)}(t)y_1(0) d\mu \text{ } dt \text{.}
\end{align}

The inner integral can also be expressed as $\mathbb{E}_{\mu}[y_1(t)y_1(0)]$

Although the equations above provide a closed-form expression for $\sigma$, they are only of limited use here as both approximating the double time integral as well as sampling from the invariant measure and integrating yield numerically poor convergence behavior.

A different approach to get $\sigma_L$, which seems to be easier to handle numerically, would be to estimate $\sigma_L^2$ using the WIP \eqref{eq:WIP} instead of the Green-Kubo formula.

We begin by solving \eqref{eq:fast_sub} for a long timespan $T$ for $N$ random initial conditions $y^{(i)}(0)$, sampled according to $\mu$. Then we fix a $t > 0$ and calculate $W_k^{(i)} = k^{-\frac{1}{2}}\int_0^{kt}y^{(i)}(\tau) d\tau$ up to $K = \frac{T}{t}$. Using that $k^{-\frac{1}{2}} \int_0^{kt}f_0(y^{(1)}(\tau)) d\tau \rightarrow_w \sigma_L W_t$, we can estimate $\hat \sigma_L^2 = \frac{1}{t} \mathrm{Cov}[W_K^{(i)}]$ for large enough $K$ respectively $T$. A result of this method can be seen in \Cref{fig:lorenz_est}. Based on this, we will use $\sigma_L^2 = 60$ for all simulations.
Note that similar methods to obtain $\sigma_L$ are discussed in \cite{Siegert1998}, where $\sigma_L$ is inferred from given stochastic time-series data via a Kramers-Moyal expansion, see also \cite{Givon2004,VandenEijnden2003}. Their methods infer the diffusion coefficient from a single trajectory, whereas the approach that we used here relies on ensemble simulations, which allows both an easy implementation for practitioners and solves the previously mentioned poor convergence.

\begin{figure}
\includegraphics[width=\linewidth]{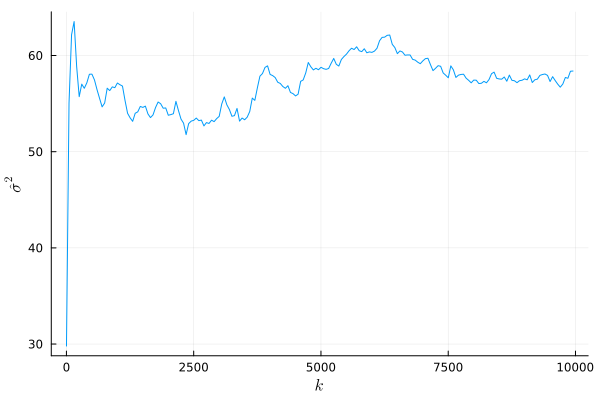}
\caption{Weak invariance principle-based estimates of $\sigma_L^2$, obtained from integrating the first component of the Lorenz-63 system. Estimated over $N = 500$ trajectories with timestep $t=1$ and lengths up to $K = T = 10000$.}
\label{fig:lorenz_est}
\end{figure}

\subsection{Kramers' law plots}
\label{sec:Kramers'plots}

We claim that we can apply \Cref{lem:melbourne} and \Cref{thm:irrev_Kramers} to the chaotically forced AMOC 3-box model and its SDE limit, respectively. In both cases, the model is locally Lipschitz, and due to the dissipativity, the relevant dynamics is restricted to a compact set, so we can apply all the above theorems requiring Lipschitz continuity. The bistable region in the hosing parameter $H$ in the HADGEM3-MM model calibration contains the interval $ \left[ 0.02, 0.2 \right ]$, a bifurcation diagram is shown in \cite{Chapman_2024}, Fig. 3.7.

To visualize the Kramers formula and to test the behavior before the limit $\epsilon \rightarrow 0$, we compute $N$ different trajectories starting in $\bar x_1$ (on-state) for different values of $\delta$ until all $N$ trajectories have tipped to the second equilibrium $\bar x_2$ (off-state) and plot the logarithm of the first time $t_i$ of entering a region sufficiently close to $\bar x_2$. We take the averages of these times for each fixed value of $\delta$ and plot them in \Cref{fig:AMOC_Kramers} panel (a). The slope of these averaged transition times (shown in red in panel (a) ) times a factor $2$ then gives the quasipotential value at the saddle\cite{Berglund_2013} $\bar V(\bar x_1,x_s)$, i.e., the communication height, as $\bar{V}(\bar x_1, \bar x_1) = 0$, and the y-axis intersection corresponds to the coefficient in \eqref{eq:irrev_kramersLogEqCtilde}. In the panels (b) to (d) in \Cref{fig:AMOC_Kramers}, which show the simulated transition times in the ODE case, we display this same red line computed from the SDE simulations for comparison between the SDE and ODE case.
Notice that Kramers' law doesn't require starting the simulations at the equilibrium, but just in its basin of attraction. We initialized the system in the on-state for $H = 0$ although the actual equilibrium is shifted as it depends on the hosing $H$. We assumed that our initial state still remained in the basin of the shifted on-state for hosing $H>0$. Note that the communication height scales quadratically when rescaling the noise matrix $A$.

\begin{figure*}[ht]
    \centering
	\hspace{-0.3\textwidth}
    \begin{subfigure}[t]{0.22\textwidth}
        \centering
        \captionsetup{width=2\linewidth}
        \includegraphics[width=2\linewidth]{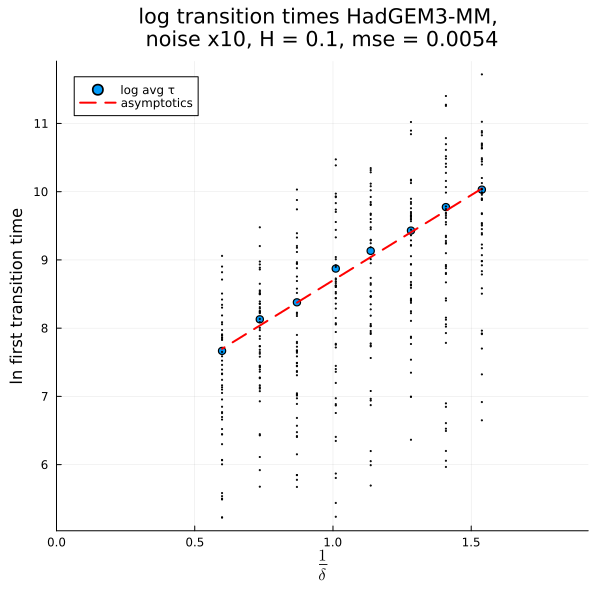}
        \caption{Kramers' law for the SDE \ref{eq:SDE}. This is the well-known case and it can be seen as the limit $\varepsilon\rightarrow0$ of the ODE case. Black dots show individual observations.}
        \label{fig:kramers_noise}
    \end{subfigure}
    \hspace{0.27\textwidth}	
    \begin{subfigure}[t]{0.22\textwidth}
        \centering
        \includegraphics[width=2\linewidth]{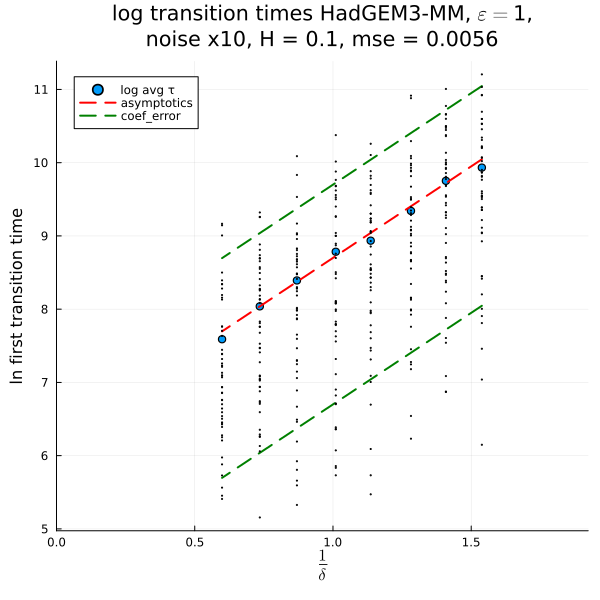}
        \caption{Chaotic ODE with $\epsilon = 1$}
        \label{fig:kramers_e01}
    \end{subfigure}
	
    \vspace{0.3cm}	
	\hspace{-0.3\textwidth}
    \begin{subfigure}[t]{0.22\textwidth}
        \centering
        \includegraphics[width=2\linewidth]{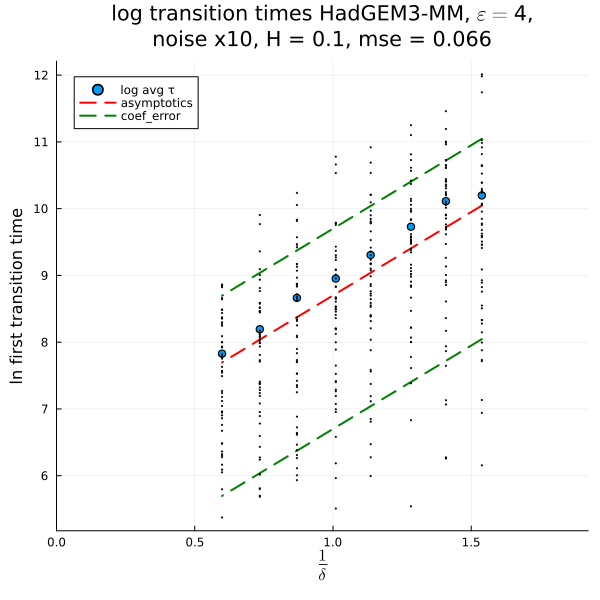}
        \caption{Chaotic ODE with $\epsilon = 4$}
        \label{fig:kramers_e04}
    \end{subfigure} 
    \hspace{0.27\textwidth}
    \begin{subfigure}[t]{0.22\textwidth}
        \centering
        \includegraphics[width=2\linewidth]{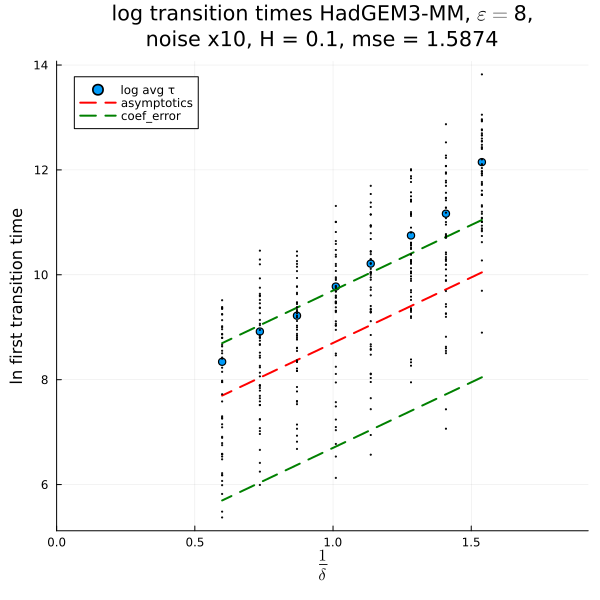}
        \caption{Chaotic ODE with $\epsilon = 8$}
        \label{fig:kramers_e08}
    \end{subfigure}

    \caption{Plots of observed tipping times in the AMOC 3-box model calibrated to match HadGEM3-MM, starting in the on-state for selected values of $\delta$. The black dots show the logarithmic waiting time until the first transition to the off-state, and their average over $50$ samples is shown by the blue dots. The red line is fitted to these averages (fit with $200$ sampled transition times per value of $\delta$) for the SDE case, but the exact same line is still shown in panels (b) to (d) for comparison to the chaotic ODE case. Note that the SDE case can be viewed as the ODE's asymptotic behavior for $\delta \rightarrow 0$. In green, the relative bounds of the error in the prefactor $\kappa^{-1} \in [e^{-2}, e^1]$. We use the noise matrix $A = 10A_{MM}$ from \cite{Chapman_2024} that was obtained via Maximum Likelihood estimation to match the HadGEM3-MM fluctuations.}
    \label{fig:AMOC_Kramers}
\end{figure*}

For small $\epsilon$, in the ODE cases, the slope is quite similar to the SDE case, while for bigger $\epsilon$, the theory of Kramers' law in the chaotic case fails. This aligns with the findings in \Cref{sec:chaotic_kramers}, where we required $\epsilon \ll \delta$. For higher values of $\epsilon$, the slope increases and below a threshold in the noise strength $\delta$, we don't see tipping at all. This behavior is as we would expect, as the chaotic perturbation is bounded by $\frac{\sqrt{2\delta}}{\sigma_L}\epsilon^{-1} \sup_y \lvert \pi_1(y) \rvert$, so at some point, it is not strong enough to be able to kick a trajectory to the other equilibrium. The deviation of the ODE case from the approximated slope in the SDE case grows exponentially in $\epsilon$ as shown in \Cref{fig:mse}. For $\epsilon=4$, the chaotic Kramers' law still agrees reasonably well with the stochastic Kramers' law in the observed range of $\delta$, despite a typical trajectory with $\epsilon=4$ not looking very similar to a noise-driven trajectory as visible in \Cref{fig:chaotic_comp}. 
In \Cref{fig:increments}, we show the distributions of the increments of the SDE and the ODE for several timescale separations $\epsilon$ and the same parameter values as in \Cref{fig:chaotic_comp}. Only for $\epsilon = 0.01$ the histograms agree well with each other, whereas for larger $\epsilon$, the distribution of increments of the ODE is visibly not close to the SDE limit. Still, in our example system, the chaotic Kramers’ law approximately holds for larger values of $\epsilon$, and thus, far beyond the stochastic limit. 
\begin{figure*}[ht]
    \centering
	\hspace{-0.3\textwidth}
    \begin{subfigure}[t]{0.22\textwidth}
        \centering
        \captionsetup{width=2\linewidth}
        \includegraphics[width=2\linewidth]{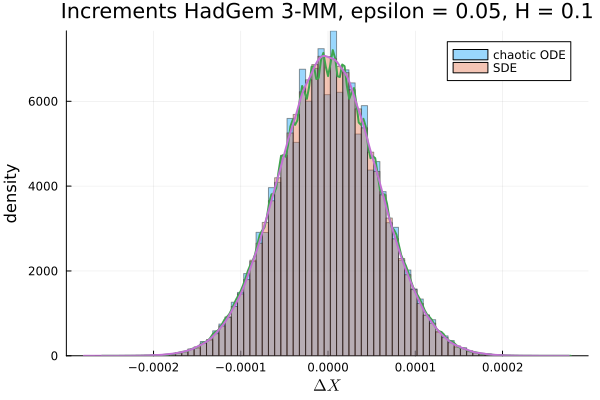}
        \label{fig:increments_e005}
    \end{subfigure}
    \hspace{0.27\textwidth}	
    \begin{subfigure}[t]{0.22\textwidth}
        \centering
        \includegraphics[width=2\linewidth]{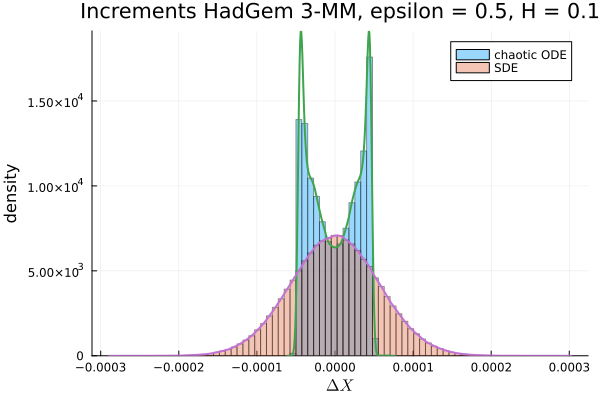}
        \label{fig:increments_e05}
    \end{subfigure}
	
    \vspace{0.3cm}	
	\hspace{-0.3\textwidth}
    \begin{subfigure}[t]{0.22\textwidth}
        \centering
        \includegraphics[width=2\linewidth]{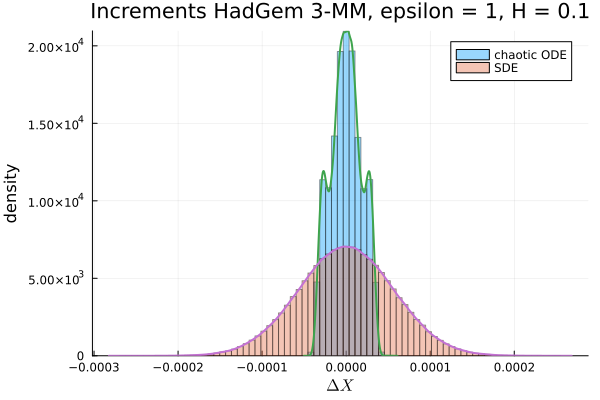}
        \label{fig:increments_e1}
    \end{subfigure}% 
    \hspace{0.27\textwidth}
    \begin{subfigure}[t]{0.22\textwidth}
        \centering
        \includegraphics[width=2\linewidth]{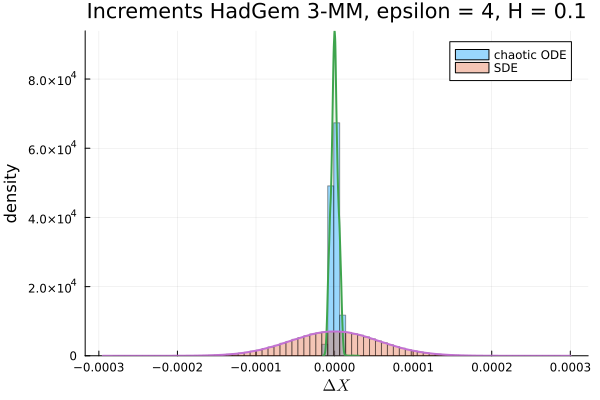}
        \label{fig:increments_e4}
    \end{subfigure}

    \caption{Plots of the increments in the AMOC 3-box model calibrated to match HadGEM3-MM, after reaching the on-state for $\delta = 1$ in comparison to the SDE limit. The increments are defined by $X(t)-X(t-\Delta t)$ with $\Delta t = 0.1$ for $t \in \left[500, 100000\right] \cap 0.1 \mathbb{N}$. 
    }
    \label{fig:increments}
\end{figure*}
Note that for small $\epsilon$, i.e., fast forcing, the response system ``sees'' the running mean of the forcing. The slower the forcing, the longer the response system sees similar forcing values, and thus, extremes along a given forcing trajectory may lead to very early or very late tipping; e.g., if the chaotic forcing spends a long time in just one wing of the Lorenz attractor.

\begin{figure}
\includegraphics[width=\linewidth]{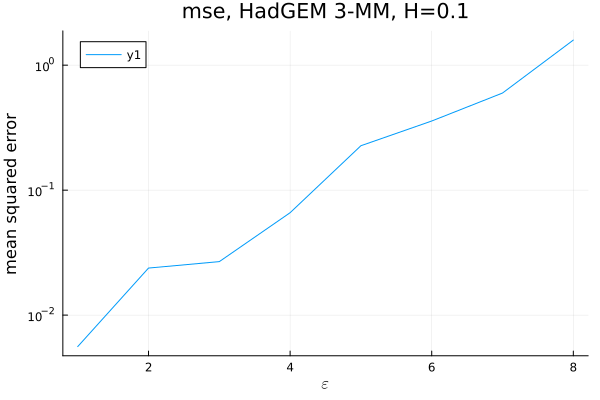}
\caption{Mean squared deviation of the mean transition time in the chaotically forced 3-box AMOC model from Kramers' law for different values of $\epsilon$}
\label{fig:mse}
\end{figure}

\begin{figure}
\includegraphics[width=\linewidth]{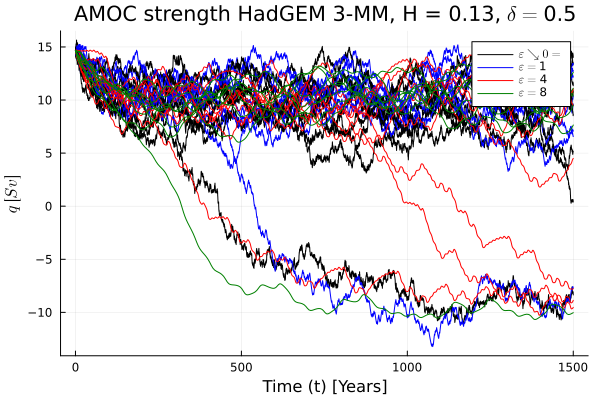}
\caption{AMOC strength of ensembles for each different value of $\epsilon$, indicating the time-scale separation. In all cases we used $\delta = 0.5$ and $A = 10A_{MM}$. The limiting SDE is shown in black.}
\label{fig:chaotic_comp}
\end{figure}

The quasipotential values in Table \ref{tab:quasi_potential_approx} show that noise- or fast chaos-induced tipping from the on- to the off-state or vice versa strongly depends on the value of the hosing parameter $H$.
A higher quasipotential value at the saddle in \Cref{tab:quasi_potential_approx} corresponds to an exponentially longer waiting time for purely noise-induced or fast chaos-induced transitions. As expected, we observe that larger hosing stabilizes the off-state and destabilizes the on-state. The latter even ceases to exist for large $H$. Similar behaviour can also be seen in other models of the AMOC \cite{Stommel_1959, VanWesten2024}.

\begin{table}[ht]
    \centering
\begin{tabular}{r|c|c|c|c}
    $\bar V(\cdot, x_s)$ & $H = 0.02$& $H = 0.05$ & $H = 0.1$ & $H = 0.15$\\
    \hline
    on state & $19$ & $13$ & $5$ & $0.4$\\
    \hline
    off state &  $0.2$& $6$ & $25$ & $50$\\
\end{tabular}
\caption{Quasipotential values of the critical saddle for HADGEM 3-MM fits of the 3-box AMOC model, starting either in the AMOC on- or off state. Values are obtained via fitting a line to the log-transition times as shown in Fig. \ref{fig:AMOC_Kramers}.}
    \label{tab:quasi_potential_approx}
\end{table}

\section{Discussion}
\label{sec:discussion}

We derived the chaotic Kramers' law for the case of a bistable chaotically forced system with the chaotic forcing being fast, but not infinitely fast, compared to the forced system. We described how it can be used to understand tipping times and verified it numerically in a reduced-order model of one of the most important potential climate tipping elements, the AMOC. The intricate numerical details that need to be considered in the case of fast chaotic forcing are discussed, and we explain how they can be handled in a numerically stable way via ensemble estimates. This work highlights both the limitations of Hasselmann's program and its surprisingly wide applicability to certain questions, as evidenced by the observed agreement between the chaotic and stochastic Kramers' scaling up to a timescale separation of $\epsilon=4$ in \Cref{fig:kramers_e08}. This is remarkable as the histograms of the increments of the ODE clearly deviate from the SDE limit for $\epsilon=4$, as shown in \Cref{fig:increments}. 
This discrepancy between the histograms may be further addressed via finite–$\epsilon$ higher-order cumulant corrections to the homogenized SDE, such as the Edgeworth expansions \cite{WoutersGottwaldSIAM2019, WoutersGottwaldRSPA019}. These corrections provide a refined description of the slow variable, yielding modified drift and diffusion terms. Thus, they would also alter the associated large deviation rate function of the effective stochastic model, consistent with the observations in \cite{Bouchet2016}. While these expansions would offer valuable theoretical insight, their practical use for estimating transition rates is limited due to their complexity and the difficulty of incorporating higher-order terms into a full LDP framework. Remarkably, our numerical results indicate that the leading-order Kramers' law obtained from the homogenized SDE already gives accurate predictions for the mean transition time even in regimes where non-Gaussian fluctuations occur. This suggests that the exponential scaling of rare-event probabilities is comparatively insensitive to moderate deviations from the CLT predicted by these higher-order expansions.

A version of Kramers' law for systems forced by multiplicative noise\cite{Rosas2016} in combination with homogenization theory results for multiplicatively forced systems \cite{Gottwald_2013} (briefly mentioned in Appendix \ref{app:MultiplicativeHomogenization}), may allow a generalization to obtain a ``multiplicative chaotic Kramers' law''.
Note that the previous literature on large deviation theory in fast-slow systems, e.g., \cite{Bouchet2016LargeDeviations, Bakhtin2003, Boerner2024} is concerned with stochastic systems showing fast-slow dynamics. This differs from our setting, as in our case, the fast-slow dynamics result in stochastic forcing and are thus no longer present in the stochastic limit. We do not consider fast-slow effects in the limiting stochastic system.
Further, note that our findings do not suggest that there is no big difference between chaotic and stochastic forcing in general. 
For example, if, for fixed forcing speed $\epsilon$, the forcing amplitude gets very small, the range of the bifurcation parameter $H$ for which tipping is possible would shrink to a chaotic tipping window\cite{Ashwin_2024, Roemer2025}, outside which tipping is not possible in response to the bounded chaotic forcing. The beginning of the transition towards this dynamic regime can be seen in the bottom-right panel in Figure~\ref{fig:kramers_e08}, where the transition times become large for cases with small forcing amplitude, and would diverge for small enough $\delta$. By contrast, unbounded noise would almost surely lead to tipping in finite time for any $H$ in the bistable regime of the response system. This shows that wrongly modelling fast chaotic dynamics as noise can lead to unreasonably small average transition times when the forcing amplitude is too small to be close to the stochastic limit.
In Lemma $\ref{lem:chaotic_kramers}$, this issue is taken care of by requiring that $\epsilon\ll\delta$, i.e., the forcing amplitude has to be large enough compared to the forcing speed. As both $\epsilon$ and $\delta$ are small, we have a double-limit problem which is usually complicated\cite{Kuehn2022}. Our derivation of the chaotic Kramers' law suggests that $\epsilon$ has to approach $0$ at least exponentially fast compared to $\delta$ for Kramers' law to hold, but a more detailed analysis of this double-limit is left for future work and may become even more interesting if also the deterministic dynamics show timescale separation\cite{Boerner2024}. Throughout the paper, we scaled the forcing with the square root of the timescale ratio $\epsilon^2$ as in \cite{Melbourne_2011} and assumed the chaotic forcing to be ergodic and bounded. Despite ergodicity and boundedness being approximately fulfilled in many real-world systems, care needs to be taken when it is unclear how close a system is to the stochastic limit. Additionally, we stress that the convergence of the chaotic system to the SDE limit is not to be understood trajectory-wise, but it is a weak convergence of measures\cite{Melbourne_2011}.
Although the chaotic Kramers' law does not reproduce the value of the quasipotential at the saddle for $\epsilon=8$ anymore, we still observe exponential scaling of transition times with the forcing strength in \Cref{fig:AMOC_Kramers} for small enough $\delta^{-1}$. This is an interesting result by itself and invites further analysis.

Note that the chaotic forcing is only a conceptual representation of natural weather variability\cite{Ashwin_2021}, unlikely to represent the actual physical system's drivers realistically, e.g., in terms of weather and especially wind forcing. Moreover, translating the hosing strength to CO$_2$ forcing is not trivial, and recall that we rescaled the noise matrix compared to \cite{Chapman_2024}. For these reasons, only limited quantitative conclusions can be drawn from our work about the real-world AMOC. However, we observe qualitative agreement with other AMOC models of different complexities \cite{Romanou_2023, VanWesten2024}, which highlights the relevance of chaotically forced AMOC models. Our work indicates that, also in purely deterministic GCMs, seemingly stochastic transitions and dynamics may be possible, which encourages the use of both reduced-order and macroscopic stochastic AMOC models, as e.g., in \cite{Chapman_2024,Chapman_2024_2,Alkhayuon_2019, Soons_Grafke_Dijkstra_2025, jacques-dumas_resilience_2024}. 

Starting the chaotically forced box model \ref{sec:3-boxModelAndLorenz} in the AMOC on-state and mimicking possible climate scenarios via the hosing parameter $H$ could involve an initial increase in $H$ with a subsequent decrease \cite{Alkhayuon_2019}. This may result in tipping to the off-state, which subsequently loses stability as a result of decreasing $H$, and consequently tipping back to the on-state in response to the nearing bifurcation, chaotic forcing, or a combination of these. Such simulations are shown in Figure \ref{fig:collapse_recovery} where we consider several different time-dependent hosing scenarios. As previously found\cite{Alkhayuon_2019}, rapidly decreasing the hosing strength after an initial overshoot over the bifurcation threshold in this box model can prevent the AMOC from relaxing to the off-state. However, the ramping speed in Figure \ref{fig:collapse_recovery} does not solely determine whether a trajectory tips or not, which highlights the influence of the chaotic forcing.
\begin{figure}
\includegraphics[width=\linewidth]{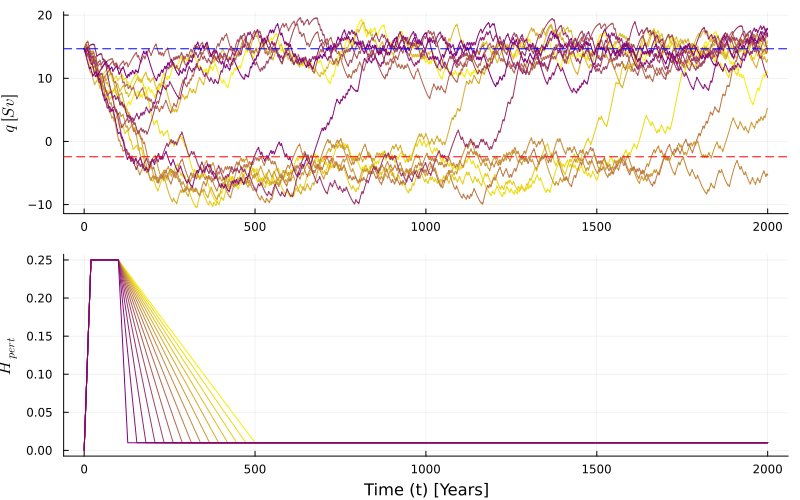}
\caption{AMOC strength (top panel) for 15 runs of the chaotically forced HADGem3-MM 3-box model with different time-dependent hosing shown in the bottom panel. In all cases we used $\epsilon = 2$, $\delta = 1.2$ and $A = 10A_{MM}$}
\label{fig:collapse_recovery}
\end{figure}
Note that the time dependence of the hosing makes the system, in contrast to the other results in this paper, non-autonomous. We still include Figure \ref{fig:collapse_recovery} here to show that a combination of ramped hosing and chaotic forcing can lead to dynamics resembling the observed AMOC collapse and recovery in the NASA-GISS model\cite{Romanou_2023}.
Thus, this setting is a candidate for explaining these NASA-GISS AMOC dynamics, although further analysis is needed to verify or discard it. Understanding the underlying dynamic mechanism of a potential real-world AMOC collapse and recovery solely using comprehensive GCMs in response to anthropogenic greenhouse gas emissions or hosing would be challenging, if not impossible. Using reduced-order models, several mechanisms were explored, including bifurcation-induced tipping \cite{Schefferetal, KuehnCT2, Alkhayuon_2019}, noise-induced tipping \cite{KuehnCT1, Chapman_2024_2}, and tipping induced by non-chaotic non-autonomous forcing \cite{Ashwin_2021, KuehnLongo}. Our work adds the potential explanation of a chaotic forcing-induced collapse to this list, which could potentially also function as a bridge between stochastic noise and non-chaotic non-autonomous forcing scenarios. In the real world, a combination of several of these mechanisms seems likely if the AMOC tips.

\section*{acknowledgments}
The authors wish to thank Ruth Chapman for helpful discussions on the choice of model parameters and the corresponding quasipotential value at the saddle; and Richard Wood for insightful discussions and comments about the 3-box-model.  
R.R., N.B., and C.K. gratefully acknowledge funding from the European Union’s Horizon 2020 research and innovation programme under the Marie Sklodowska-Curie Grant Agreement No. 956170 (CriticalEarth). NB acknowledges funding by the Volkswagen foundation. This is ClimTip contribution \#X; the ClimTip project has received funding from the European Union's Horizon Europe research and innovation programme under grant agreement No. 101137601.  

\section*{author declarations} 
The authors declare that they have no competing interests or conflicts of interest.
The individual authors contributed the following:
J.D.: literature review; significant ideas, especially for the proof of the chaotic Kramers' law; all simulations; all plots; most of the writing of Sections 2,3.
R.R.: literature review; idea to look at a model forced with fast chaos and at Kramers' law in the chaotic case; significant ideas in Sections 2,3, and contributions to writing them; writing of the other sections.
N.B.: Idea to look at tipping in the AMOC-box-model to better understand NASA-GISS simulations; ideas on the physical interpretation.
C.K.: main ideas on how to prove the chaotic Kramers' law; contributions to the mathematics in Section 2

\section*{data availability statement}
The codes to run the model and generate the figures are available at: \href{https://github.com/raphael-roemer/chaotic_Kramers_law-AMOC_tipping}{https://github.com/raphael-roemer/chaotic\_Kramers\_law-AMOC\_tipping}

\appendix
\renewcommand{\theequation}{\Alph{section}\arabic{equation}}

\section{Homogenization for Multiplicative Forcing}
\label{app:MultiplicativeHomogenization}
\setcounter{equation}{0}

A very similar statement to \Cref{lem:melbourne} is true for the case of multiplicative (i.e., $x$-dependent) noise. The slow deterministic subsystem is then of the form
\begin{align}
\dot{x}^{(\epsilon)} = \epsilon^{-1}h(x^{(\epsilon)})f_0(y^{(\epsilon)}) + f(x^{(\epsilon)})
\end{align} and the limit SDE will have a diffusion coefficient of the form $\sqrt{\Sigma} h(X)$. However, the multidimensional version of this result requires restrictive assumptions on $h$. The statement was proven in \cite{Gottwald_2013}, where one also finds a discrete version of \Cref{lem:melbourne}

\section{Deriving Kramers' law from large deviations theory}
\label{app:LargeDevToKramers}
\setcounter{equation}{0}
\subsection{Large Deviation Principle (LDP)}

The following result on large deviations in dynamical systems is attributed to \cite{Freidlin_wentzell_2012}; we refer to the more convenient formulation in \cite{Dembo_2009}, proceeding Theorem 5.6.7. An introduction to the mathematical links between minimum action, quasi-potentials, and large deviations can be found in \cite{Schapers_2019}.

Let $\Gamma \subset \mathcal{H}_1$ be a regular set of paths with respect to $I$, i.e., the minimization of $I$ over the interior $\Gamma^\circ$ equals the minimization over the closure $\bar \Gamma$, so $\inf_{\varphi \in \Gamma^\circ}I(\varphi) = \inf_{x \in \bar \Gamma}I(x)$.
\begin{theorem}[Anisotropic Freidlin-Wentzell Theorem] \label{thm:transformed_FWT}
Consider an SDE with additive noise, i.e. $dX_{(\delta)} = F(X_{(\delta)})dt + \sqrt{2\delta}AdW$ with $A \in \mathbb R^{n \times n} $ non-degenerate, $F$ Lipschitz and bounded. Then the distributions of paths $X_{(\delta)}\in \mathcal{H}_1$, i.e., the family of probability measures $\mathbb P [X_{(\delta)} \in B]$ for measurable sets of paths $B\subset\mathcal{H}_1$, satisfies a large deviation principle (LDP) \cite{dembo_large_2010}:
\begin{align}
\label{eq:LDP}
- \inf_{x \in \Gamma}I(x) = \lim_{\delta \rightarrow 0} \delta \log \mathbb P [X_{(\delta)} \in \Gamma] 
\end{align}
with the Freidlin-Wentzell action \eqref{eq:FW-action} as the good rate function.
\end{theorem}
We will also write $\mathbb P[X_{(\delta)} \in \Gamma]\simeq e^{\frac{-\inf I(\Gamma)}{\delta}}$. This means that in the limit $\delta\rightarrow0$, paths minimising the good rate function $I$ become exponentially more likely than other paths. The continuous minimising path from $\bar x_1$ to some $z \in D$ is called the {\em instanton} $\psi$ from $\bar x_1$ to $z$. For any point $z\in D$ it holds that $I(\psi)= \bar V(\bar x_1, z)$ by definition of $\bar V$.

Note that the quasipotential with respect to the good rate function (\ref{eq:FW-action}) is Lipschitz continuous in the second variable, which follows from Lemma 2.3 in \cite{Freidlin_wentzell_2012}.

\subsection{From LDP to Kramers' law} \label{sec:ldp_to_kramers}
Next, we indicate how to derive the irreversible Kramers formula for $\dot x = F(x) dt + \sqrt{2\delta}A dW_t$ using a large deviation principle. This is the basis for subsequently deriving the chaotic Kramers' law. 

Assume a deterministic bistable ODE $\dot x = F(x)$ with a unique critical saddle $x_s$ (a saddle with a single unstable eigendirection) on the basin boundary between two asymptotically stable equilibrial $\bar x_1$ and $\bar x_2$ and let $D$ be in the basin of attraction of $\bar x_1$. By Theorem 2.1 in \cite{Freidlin_wentzell_2012}, the location of exit of the SDE $d x = F(x) dt + \sqrt{2\delta}dW_t$ from $D$ concentrates at a single point for $\delta \rightarrow 0$ if the quasipotential attains its minimum over $\bar D$ at this point. Given the right choice of $D$, this point will be arbitrarily close to the critical saddle. We hypothesize that this holds true in the anisotropic case as well. 

We can use that the first exit time of an SDE solution satisfies a large deviation principle to deduce that for the right choice of $D$
\begin{equation}
\begin{aligned}  \label{eq:exit_time_D}
\lim_{\delta \rightarrow 0}2\delta \log \mathbb{E}_{x_0}[\tau_{1, 2}] = \lim_{\delta \rightarrow 0}2\delta \log \mathbb{E}_{\bar x_1}[\tau_{D}] \\= \inf_{z\in \partial D}\bar V(\bar x_1, z) = \bar V(\bar x_1, x_s)\text{.} 
\end{aligned}
\end{equation}
We briefly summarise the sketch of the proof\cite{Berglund_2013} of Equation (\ref{eq:exit_time_D}). The authors\cite{Berglund_2013} argue that the times for reaching a small neighbourhood of $\bar x_1$ starting from some $x_0 \in D$ and for reaching a small neighborhood of $\bar x_2$ after leaving $D$ at $x_s$ are negligible, thus the first equality in Equation (\ref{eq:exit_time_D}). Using the LDP (\ref{eq:LDP}) for a set of paths $\Gamma$ from $\bar x_1$ to $x_s$ gives $\mathbb{P}_{\bar x_1}[X_{(\delta)} \in \Gamma] \simeq e^{-\frac{I(\psi)}{2 \delta}} = e^{-\frac{\bar V(\bar x_1, x_s)}{2\delta}}:=p$. 
A solution $X_{(\delta)}$ can be divided into multiple shorter trajectories of some duration $T_1 > 0$, and using the \textit{Markov property} of $X_{(\delta)}$, each of these shorter trajectories can then be regarded as ``starting a new trial'' to exit $D$. The number of attempts follows a geometric distribution with expected value $\frac{1}{p} = e^{\frac{\bar V(\bar x_1, x_s)}{2\delta}}$, and using that the location of exit concentrates at $x_s$ for $\delta$ small enough\cite{Freidlin_wentzell_2012, Boerner2024}, we can find an upper bound of the probability of leaving $D$ in some time $T_1 > 0$: 
\begin{align*}
E_{\bar x_1}[\tau_{D}] \lesssim T_1 e^{\frac{\bar V(\bar x_1, x_s)}{2\delta}} \text{.}
\end{align*} 
Using a more refined reasoning (Theorem 5.7.11 (a) in \cite{Dembo_2009}), one can prove a lower bound of the form 
\begin{align*}
E_{\bar x_1}[\tau_{D}] \gtrsim c e^{\frac{\bar V(\bar x_1 x_s)}{2\delta}}
\end{align*} for some $c < 0$. By $\lesssim$ and $\gtrsim$ we mean inequality up to a small perturbation of $\bar V$ and for small enough $\delta$.

The approximation errors $c$ and $T_1$ vanish in the logarithmic limit \eqref{eq:exit_time_D}, and we find Kramers' law \eqref{eq:irrev_kramersLogEq}. Note that the prefactor is not derived. Without this prefactor, Kramers' law is also known as Arrhenius' Law.

\section{Parameters}
\label{sec:Appendix}
\setcounter{equation}{0}

The parameters\cite{Alkhayuon_2019, Wood_2019, Chapman_2024_2} used for the AMOC box model in Section \ref{sec:3-boxModelAndLorenz} are given in Table \ref{tab:parameters} 

\begin{table}[h!]
\centering
\begin{tabular}{|l|l|l|}
\hline
\textbf{parameter} & \textbf{value} & \textbf{unit}\\
\hline
$V_{N}$ & $4.192 \cdot 10^{16}$ & $m^3$\\
$V_{T}$ & $4.191 \cdot 10^{16}$ & $m^3$\\
$V_{S}$ & $13.26 \cdot 10^{16}$ & $m^3$\\
$V_{IP}$ & $16.95 \cdot 10^{16}$ & $m^3$\\
$V_{B}$ & $96.76 \cdot 10^{16}$ & $m^3$\\
$A_{N}$ & $0.9841 \cdot 10^{6}$ & \\
$A_{T}$ & $-0.1853 \cdot 10^{6}$ & \\
$F_{N}$ & $0.2799 \cdot 10^{6}$ & $Sv$\\
$F_{T}$ & $-0.7920 \cdot 10^{6}$ & $Sv$\\
$T_{S}$ & $5.349 $ & $ ^\circ C$\\
$T_{0}$ & $4.514 $ & $ ^\circ C$\\
$\mu$ & $0.0 \cdot 10^{-8}$ &  $ ^\circ C \text{ } m^{-3}s$\\
$\lambda$ & $2.328 \cdot 10^{7}$ & $m^6 \text{ } kg^{-1} \text{ } s^{-1}$\\
$K_{N}$ & $4.73 \cdot 10^{6}$ & $Sv$ \\
$K_{S}$ & $7.68 \cdot 10^{6}$ & $Sv$\\
$\gamma$ & $0.58$ &  \\
$\alpha$ & $0.12$ & $kg \text{ } m^{-3} \text{ } (^\circ C)^{-1}$\\ %kg m^-3 C^-1
$\beta$ & $790.0 $ &  $kg \text{ } m^{-3}$ \\% kg m^-3 
$Y$ & $3.15 \cdot 10^{7}$ & \\ % sec/year
$S_0$ & $0.035$ & ppt $\cdot 10^{3}$\\
$S_N^{(eq)}$ &  $0.034912$ & ppt $\cdot 10^{3}$\\
$S_T^{(eq)}$ & $0.035435$ & ppt $\cdot 10^{3}$\\
$S_S^{(eq)}$ & $0.034427$ & ppt $\cdot 10^{3}$\\
$S_{IP}^{(eq)}$ & $0.034668$ & ppt $\cdot 10^{3}$\\
$S_B^{(eq)}$ & $0.034538$ & ppt $\cdot 10^{3}$\\

% Other parameter
$\sigma_L$ & $\sqrt{60}$ & \\
\hline
\end{tabular}
\caption{Parameters used in the HadGEM 3 calibration of the 3-box AMOC model \cite{Alkhayuon_2019, Wood_2019, Chapman_2024_2} and the chaotic limit. The units ppt (parts per thousand) and psu (practical salinity unit) are used interchangeably in this context. $Y$ represents the number of seconds per year.}
\label{tab:parameters}
\end{table}
For $i \in \{N, S\}$ define $C_i = \frac{Y}{V_i}$. Note that $Sv$ (Sverdrup) is defined as $1Sv =10^6m^3\text{ }s^{-1}$.
The total salt content $C$ can be computed as 
\begin{align*}
    C = V_N  S_N^{(eq)} + V_T  S_T^{(eq)} + V_S  S_S^{(eq)} + V_IP  {S_{IP}}^{(eq)} + V_B  S_B^{(eq)} \text{.}
\end{align*}
Also $\tilde S_B = 100(S_B^{(eq)}-S_0)$ and $\tilde S_S = 100(S_S^{(eq)}-S_0)$.

The entries of the diffusion coefficient matrix \par
$A_{MM} = \begin{pmatrix} 0.1263 & 0 \\ -0.0869 & 0.1088 \end{pmatrix}$
have unit \cite{Chapman_2024_2} $\text{ppt (year)}^{-\frac{1}{2}} \cdot 10^{-2}$. 
So in order to match the unit $\text{ppt} \cdot 10^1$ of the model variables $\tilde S_N$ and $\tilde S_T$, we find the rescaled noise coefficient matrix to be $\tilde A_{MM} = A_{MM} \cdot 10^{-3}$. For the numerical simulations, $10 \tilde A_{MM}$ was used in order to reduce integration time for observed tipping. The values of the quasipotential at the saddle without this additional scaling can be easily calculated as the quasipotential scales quadratically upon rescaling the diffusion coefficient matrix. 

For the Lorenz 63 system, we used the classical parameter setting $\tilde{\rho} = 28$, $\tilde{\sigma} = 10$, and $\tilde{\beta} = \frac{8}{3}$.

%\nocite{*}
\bibliography{literature}% Produces the bibliography via BibTeX.

\end{document}